\documentclass[journal,twoside,web]{ieeecolor}

\usepackage{generic}
\usepackage{cite}
\usepackage{amsmath,amssymb,amsfonts}
\usepackage{graphicx}
\usepackage{hyperref}
\hypersetup{hidelinks}
\usepackage{textcomp}

\let\labelindent\IEEElabelindent
\usepackage{enumitem}
\usepackage{flushend}

\hypersetup{colorlinks=true,linkcolor=black,citecolor=blue}
\graphicspath{{Figs/}}

\usepackage{xcolor}
\definecolor{myred}{rgb}{0,0,0}\def\myred#1{{\textcolor{myred}{#1}}}
\definecolor{myblue}{rgb}{0,0.32,0.7}

\usepackage{tikz}
\usetikzlibrary{arrows}

\tikzset{
    mynode/.style={rectangle,rounded corners,draw=black, 
      thick, inner sep=.5em, minimum size=2em, text centered},
    myarrow/.style={->, >=latex', shorten >=1pt, thick},
    }

\newtheorem{theorem}{Theorem}
\newtheorem{remark}{Remark}
\newtheorem{proposition}{Proposition}

\def\BibTeX{{\rm B\kern-.05em{\sc i\kern-.025em b}\kern-.08em
    T\kern-.1667em\lower.7ex\hbox{E}\kern-.125emX}}
\begin{document}

\title{Leader-Follower Formation and Tracking Control of Underactuated Surface Vessels}
\author{Bo Wang and Antonio Lor{\'\i}a
\thanks{This work was supported in part by GSoE at CCNY, and in part by the PSC-CUNY Award, jointly funded by
The Professional Staff Congress and The City University of New York. 
\textit{(Corresponding author: Bo Wang.)}}
\thanks{Bo Wang is with the Department of Mechanical Engineering, The City College of New York, New York, NY 10031 USA (e-mail: bwang1@ccny.cuny.edu).}
\thanks{Antonio Lor{\'\i}a is with the Laboratoire des signaux et syst{\`e}mes (L2S), CNRS, 91190 Gif-sur-Yvette, France (e-mail: antonio.loria@cnrs.fr).}
}

\maketitle

\begin{abstract}
This Technical Note presents a simple control approach for global trajectory tracking and formation control of underactuated surface vessels equipped with only two propellers. The control approach exploits the inherent cascaded structure of the vehicle dynamics and is divided into control designs at the kinematics and kinetics levels. A controller with a low-gain feature is designed at the kinematics level by incorporating the cascaded system method, persistency of excitation, and the small-gain theorem. Furthermore, a PD+ controller is designed to achieve the velocity tracking at the kinetics level. The proposed control laws are partially linear and saturated linear and easy to implement. Based on a leader-follower scheme, our control approach applies to the formation tracking control problem of multi-vehicle systems under a directed spanning tree topology. \color{myred}Our main results guarantee uniform global asymptotic stability for the closed-loop system, which implies robustness with respect to bounded disturbances in the sense of Malkin’s total stability, also known as local input-to-state stability.
\color{black}
\end{abstract}


\section{Introduction}\label{sec:introduction}

\IEEEPARstart{M}{otion} control of underactuated surface vessels has received much consideration over the last two decades due to its intrinsic nonlinear properties and practical applications in rescue, search, exploration, and reconnaissance missions. \myred{Unlike fully-actuated systems, underactuated systems cannot be commanded to follow arbitrary trajectories.} Furthermore, due to the underactuation, 
a surface vessel equipped with only two propellers cannot be stabilized by continuous time-invariant feedback \cite{pettersen1996exponential}.  Consequently, trajectory tracking and set-point stabilization are usually studied as two separate problems in the context of underactuated surface vessels. 

Trajectory tracking control of underactuated surface vessels has been thoroughly studied since the late 1990s\,---\,see, e.g., \cite{pettersen2001underactuated,pettersen1998tracking,jiang2002global,do2002underactuated,lefeber2003tracking,lefeber2000tracking}. Early results primarily focus on local tracking \cite{pettersen2001underactuated} and semi-global tracking \cite{pettersen1998tracking}. The first global tracking result utilizing the cascaded system approach for nonholonomic systems is reported in \cite{panteley1998exponential}. Later, using a similar cascaded system approach, the global tracking problem for underactuated surface vessels was solved in \cite{lefeber2000tracking}. Subsequently, several nonlinear control methods have been proposed in the literature based on the backstepping technique; for instance, in \cite{jiang2002global,do2002underactuated,ghommam2010global}. {\color{myred}One characteristic of the controllers obtained through backstepping is that they are highly complex, making implementation and tuning challenging.} Several control methods based on the sliding mode technique have been presented to solve the robust tracking control problem for underactuated surface vessels, for example, \cite{yu2012sliding} and \cite{ashrafiuon2016trajectory}. Using the cascaded system approach, the global path-following problem for underactuated surface vessels was solved in \cite{belleter2019observer}.

Leader–follower formation tracking control is a natural extension of the classical trajectory tracking control problem for multi-agent systems. The formation tracking control involves coordinating all the agents to achieve a predefined geometric configuration through local interactions and to follow a designated swarm leader. {\color{myred}However, one essential difference between trajectory tracking and formation tracking lies in the fact that formation tracking is inherently \textit{distributed}, necessitating that each agent relies on \textit{relative} position measurements. In other words, agents lack GPS measurements, rendering control laws based on absolute positions unfeasible \cite{oh2015survey}.} The leader-follower formation tracking control problem for underactuated surface vessels has been studied, e.g., in \cite{jin2016fault,ghommam2017adaptive,lu2020adaptive,wang2022robust,10076261}. In \cite{jin2016fault}, a barrier Lyapunov function-based formation control scheme is proposed to deal with line-of-sight range and bearing constraints for surface vessels. Using the adaptive backstepping method, formation tracking control with asymmetric range and bearing constraints is considered in \cite{ghommam2017adaptive}. An output feedback formation controller has been proposed in \cite{lu2020adaptive} to solve the formation tracking problem, ensuring error convergence in a practical sense. In \cite{wang2022robust}, a robust formation control and obstacle avoidance scheme for underactuated surface vessels is proposed for underactuated surface vessels using the super-twisting control technique. \color{myred}In \cite{10076261} bearing measurements are used in a controller that guarantees robustness with respect to bounded disturbances.\color{black}

In this Technical Note, we address the problem of leader-follower formation tracking control of underactuated surface vessels. That is, each surface vessel follows one leader, and only one swarm leader vessel has the information of the reference trajectory. The control design is based on the cascaded system approach, passivity-based control, small-gain theorem, and the construction of ISS Lyapunov functions.

\color{myred}      Our controller benefits from the inherent cascaded structure of the vessel dynamics, resulting in a very natural and simple control structure---the control laws at the kinematics level are \textit{linear} and \textit{saturated linear}, and the control laws at the kinetics level are of the proportional-derivative plus dynamic compensation (PD+) type \cite{paden1988globally}.

Our main results guarantee uniform global asymptotic stability (UGAS) of the origin for the closed-loop system, as per \cite[Def. 4.4]{Khalil2002}. The importance of this property cannot be overestimated. In the context of tracking control problems, in which case the closed-loop system is non-autonomous, only this type of stability ensures robustness with respect to perturbations in the sense of Malkin's total stability \cite[Chap. 2, Sect. 4]{ROUHABLAL}, which is  also known as local input-to-state stability. Essentially, UGAS implies the existence of a converse strict Lyapunov function that may be used to establish the robustness property. For instance, a proof may be inferred from \cite[Def. 4.4, Lemma 4.5, Thm. 4.9, Thm. 4.14]{Khalil2002} and \cite[Chap. II, Thm. 4.4]{ROUHABLAL}, but the fact is well-documented in the literature---see \cite{STAB-SURVEY-IFAC} for a detailed discussion and references and \cite{Teel20062219} for examples of non-uniformly globally asymptotically stable non-autonomous systems that may be destabilized by ``small'' bounded disturbances. 

\color{black}Our control design is based on the methodologies described in \cite{maghenem2017cascades} and \cite{maghenem2017global} for nonholonomic vehicles. To the best of the authors' knowledge, however, this marks the first application of a PD+ controller in formation control for \textit{underactuated} surface vessels. \color{myred}The problem of {\it compensating} for disturbances, {e.g.,} via the introduction of discontinuous control terms as in \cite{10076261}, remains out of scope of this Note.\color{black} 

The remainder of this Technical Note is organized as follows: In Section \ref{sec:main} we present problem formulation and our main results. Section \ref{sec:simulation} contains simulation results that illustrate the practical application of our theoretical findings. Finally, we wrap up the paper with concluding remarks in Section \ref{sec:conclusion}.

\section{Problem Formulation and Main Results}\label{sec:main}

For clarity of exposition, we present first a result on leader-follower tracking control ({\it i.e.,} limited to two surface vessels) and describe the control approach. Then, we present a result applicable to multiple leader-follower surface vessels operating over a directed spanning tree topology. 

\subsection{Cascades-Based Trajectory Tracking Control}

Let us consider the Lagrangian model of a surface vessel with only two propellers that provide the surge force and yaw torque, which is given by the equations---see, e.g., \cite{jiang2002global,do2002underactuated,lefeber2003tracking}, 
\begin{subequations}
  \label{eq:1}
  \begin{eqnarray}
     \label{eq:1a} & \dot{q}=J(q)v &\\
     \label{eq:1b} & M\dot{v}+C(v)v+Dv=G\tau. & 
  \end{eqnarray}
\end{subequations}
In the latter $q=[x~y~\theta]^\top$ is the configuration of the surface vessel containing the Cartesian coordinates $(x,y)$ and the orientation $\theta$ of the surface vessel in the fixed inertia frame; $v=[v_x~ v_y~ \omega]^\top$ is the generalized velocity vector consisting of the linear velocity $(v_x,v_y)$ and the angular velocity $\omega$ in the body-fixed frame; and $\tau=[\tau_x~\tau_\omega]^\top$ is the control input vector consisting of the surge force and the yaw torque. $J(\cdot)$ is the rotation matrix 
\begin{equation*}
    J(q)=\begin{bmatrix}
        \cos\theta & -\sin\theta & 0\\
        \sin\theta & \cos\theta & 0\\
        0 & 0 & 1
    \end{bmatrix};
\end{equation*}
$M=\operatorname{diag}\{m_{11},m_{22},m_{33}\}>0$ is the inertia matrix; $D=\operatorname{diag}\{d_{11},d_{22},d_{33}\}>0$ is the damping matrix. Finally, the Coriolis and centrifugal matrix $C(\cdot)$, and the input matrix $G$, are respectively given by
\[
    C(v)=\begin{bmatrix}
        0 & 0 & -m_{22}v_y\\
        0 & 0 & \phantom{-}m_{11}v_x\\
        m_{22}v_y & -m_{11}v_x & \phantom{-}0
    \end{bmatrix},
\quad     G=\begin{bmatrix}
        1 & 0 \\
        0 & 0 \\
        0 & 1
    \end{bmatrix}.
\]

\begin{subequations}\label{eq:5}
\myred{Since underactuated systems cannot be commanded to follow arbitrary trajectories, the trajectory tracking control problem consists in following a fictitious surface vessel with the same dynamical model}
  \begin{eqnarray}
     &\dot{q}_d=J(q_d)v_d &\label{eq:5a}\\
     &M\dot{v_d}+C(v_d)v_d+Dv_d=G\tau_d&\label{eq:5b}
  \end{eqnarray}
\end{subequations}
with bounded force reference $\tau_d=[\tau_{xd}~\tau_{\omega d}]^\top$ and coordinates $q_d=[x_d~y_d~\theta_d]^\top$ and $v_d=[v_{xd}~v_{yd}~\omega_d]^\top$. The trajectory tracking control objective is to steer the error $q-q_d$ to zero. Then, to address this problem as one of stabilization of an equilibrium, we define the body-fixed-frame position errors
\begin{align}\label{eq:error}
    \begin{bmatrix}
        e_x \\ e_y \\ e_\theta
    \end{bmatrix}
    =
    \begin{bmatrix}
        \phantom{-}\cos\theta\ & \sin\theta\ & 0 \\
        -\sin\theta\ & \cos\theta\ & 0\\
        \ 0 & 0 & 1
    \end{bmatrix}
    \begin{bmatrix}
        x-x_d \\ y-y_d \\ \theta-\theta_d
    \end{bmatrix}.
\end{align}
Hence, in the new coordinates, the error dynamics between the virtual reference vessel and the follower vessel become
\begin{subequations}
    \label{eq:error-kinematics}
  \begin{eqnarray}
    \label{eq:error-kinematics:a}
    \dot{e}_x&=&\omega e_y +v_x -v_{xd}\cos e_\theta - v_{yd}\sin e_\theta \\
    \label{eq:error-kinematics:b}
    \dot{e}_y&=&-\omega e_x +v_y +v_{xd}\sin e_\theta - v_{yd}\cos e_\theta \\
    \label{eq:error-kinematics:c}
    \dot{e}_\theta&=&\omega-\omega_d
  \end{eqnarray}
\end{subequations}
and the trajectory control problem is transformed into stabilizing the origin for the error dynamics (\ref{eq:error-kinematics}).

As in \cite{maghenem2017global}, we address this problem by separating the stabilization tasks at the kinematics and the kinetics levels. That is, we design the virtual control laws $v^*=[v_x^*~ v_y^*~\omega^*]^\top$ to stabilize the origin for the error kinematics (\ref{eq:error-kinematics}), while for the Lagrangian dynamics (\ref{eq:1b}) we design a PD+ controller to ensure that $v\to v^*$. 

\subsubsection{Underactuation constraints} Due to the underactuation, there is no external force that can directly control $v_y$. In other words, the virtual control $v_y^*$ cannot be arbitrarily assigned. On the other hand, it can be seen from the Lagrangian dynamics (\ref{eq:1b}) that $v_y$ depends on $v_x$ and $\omega$. Hence, given the virtual control laws $v_x^*$ and $\omega^*$, we define $v_y^*$ dynamically, {\it i.e.}, by solving the second equation in \eqref{eq:1b}. That is, we set
\begin{equation}\label{eq:constraint}
    \dot{v}_y^*=-\frac{m_{11}}{m_{22}}v_x^*\omega^*-\frac{d_{22}}{m_{22}}v_y^*, \quad v_y^*(0)=v_{yd}(0).
\end{equation}
\myred{In other words,  \eqref{eq:constraint} is an artificial constraint imposed by the system equation.} This acceleration constraint is a feasibility condition for the trajectory tracking control problem for underactuated surface vessels \cite{wang2022robust}. 

\subsubsection{Velocity tracking} Let $v_x^*$ and $\omega^*$ be given virtual control laws, and let $v_y^*$ be generated by (\ref{eq:constraint}). We design a controller such that the velocity error $\tilde{v}:=[\tilde{v}_x~\tilde{v}_y~\tilde{\omega}]^\top:=v-v^*\to 0$. \color{myred}To that end, first we recall that for \textit{fully-actuated} Lagrangian systems, this velocity tracking task can be easily achieved by a simple PD+ controller \cite{paden1988globally}:
\begin{equation}\label{eq:fully-actuated-control}
    u=M\dot{v}^*+C(v)v^*+Dv^*-K_d \tilde{v}.
\end{equation}
Indeed, the fully-actuated Lagrangian dynamics
\begin{equation}\label{eq:fully-actuated}
    M\dot{v}+C(v)v+Dv=u, \qquad u=[u_x~u_y~u_\omega]^\top \in\mathbb{R}^3,\quad \
\end{equation}
 in closed loop with the controller \eqref{eq:fully-actuated-control}, yields the closed-loop system
\begin{equation}\label{eq:fully-actuated-closed-loop}
    M\dot{\tilde{v}}+C(v)\tilde{v}+[D+K_d]\tilde{v}=0.
\end{equation}
Now, the total derivative of $V(\tilde{v}):=\frac{1}{2}\tilde{v}^\top M \tilde{v}$ along the trajectories of \eqref{eq:fully-actuated-closed-loop} yields
\begin{equation}
    \dot{V}_{(\ref{eq:fully-actuated-closed-loop})}=-\tilde{v}^\top [D+K_d]\tilde{v}<0
\end{equation}
for any $K_d=\operatorname{diag}\{k_{dx},k_{dy},k_{d\omega}\}\ge 0$ and, since $D=D^\top > 0$, the origin of the closed-loop system (\ref{eq:fully-actuated-closed-loop}) is uniformly globally exponentially stable (UGES) if $K_d\ge 0$. Note that we also have $\tilde{v}\in\mathcal{L}_2$. 

Thus, it is desirable to apply \eqref{eq:fully-actuated-control} to the system \eqref{eq:1b} by setting $G\tau = u$, but since \eqref{eq:fully-actuated-control} is underactuated, the second equation in $G\tau = u$ necessarily reads  
\[
    u_y = m_{22}\dot{v}_y^*+m_{11}v_x\omega^*+d_{22}v_y^* -k_{dy} \tilde v_y,
\]
and, in view of \eqref{eq:constraint}, we have 
\[
    u_y = m_{11} \tilde{v}_x\omega^* -k_{dy} \tilde v_y,
\]
which constitutes a ``perturbation'' on the closed-loop dynamics. Therefore, we set $k_{dy} = 0$. Thus, the control law \eqref{eq:fully-actuated-control}, originally designed for fully-actuated systems, may be implemented on \eqref{eq:1b} by defining 
\begin{equation}\label{eq:underactuated-control}
    \tau=G^\dagger(M\dot{v}^*+C(v)v^*+Dv^*-K_d \tilde{v}),
\end{equation}
where $G^\dagger$ corresponds to the left pseudoinverse of $G$, $v_y^*$ is generated by \eqref{eq:24},  $v_x^*$ and $\omega^*$ are arbitrary (smooth and bounded) and $k_{dy} = 0$. \color{black}In this case, the closed-loop system in the error coordinates is given by
\begin{equation}\label{eq:underactuated-closed-loop}
    M\dot{\tilde{v}}+C(v)\tilde{v}+[D+K_d]\tilde{v}=\begin{bmatrix}
        0\\ -m_{11} \tilde{v}_x\omega^*\\0
    \end{bmatrix}
\end{equation}
and $V(\tilde{v}):=\frac{1}{2}\tilde{v}^\top M \tilde{v}$ satisfies 
\begin{equation}
    \dot{V}_{(\ref{eq:underactuated-closed-loop})}=-\tilde{v}^\top [D+K_d]\tilde{v} -m_{11}\omega^*\tilde{v}_x\tilde{v}_y
\end{equation}
along trajectories of (\ref{eq:underactuated-closed-loop}). Then, if $\omega^*$ is bounded, {\it i.e.}, if  $|\omega^*| \le \omega_M$, then \myred{using Young's inequality on the cross-term} we obtain
\begin{equation}
    \dot{V}_{(\ref{eq:underactuated-closed-loop})}\le -\frac{1}{2} \tilde{v}^\top D \tilde{v}<0, \quad \forall \tilde v \neq 0,
\end{equation}
provided that $k_{dx}$ is selected as $k_{dx}\ge {m_{11}^2\omega_M^2}/{(2d_{22})}$. 

We draw the following statement from previous analysis.
\begin{proposition}[Velocity tracking]\label{prop:1}
    Consider the Lagrangian dynamics (\ref{eq:1b}) in closed loop with the control law (\ref{eq:underactuated-control}), where $t\mapsto v^*(t)$ is defined on $[t_\circ, \infty)$ for any $t_\circ \geq 0$ and satisfies  (\ref{eq:constraint}), and   $K_d :=\operatorname{diag}\{k_{dx},0,k_{d\omega}\}$, with  $k_{dx}\ge {m_{11}^2\omega_M^2}/{(2d_{22})}$ and  $k_{d\omega}>0$. Then, the origin of the closed-loop error system (\ref{eq:underactuated-closed-loop}) is globally exponentially stable, uniformly in the initial conditions $\tilde v(0)$ and $\omega^*$ satisfying  $|\omega^*|_\infty \le \omega_M$. 
\end{proposition}

{\color{myred}
\begin{remark}
    Achieving velocity tracking for fully-actuated Lagrangian systems using PD+ controllers has been well studied since \cite{slotine1987adaptive,paden1988globally}. However, to the best of the authors' knowledge, our approach marks the first application of a PD+ controller to \textit{underactuated} Lagrangian systems. Despite the apparent similarity of \eqref{eq:1b} and \eqref{eq:fully-actuated}, extending similar results to the underactuated case is far from straightforward.
\end{remark}
}

\begin{remark}
    In Proposition \ref{prop:1}, we make an abuse of notation in writing `$v^*(t)$'. Indeed, $v^*=[v_x^*~ v_y^*~\omega^*]^\top$, where $v_x^*$ and $\omega^*$ are virtual control laws to be defined in function of $t$ and $e$, and $v_y^*$ is defined dynamically via \eqref{eq:constraint}. Therefore, in the statement of Proposition \ref{prop:1}, $v^*(t)$ must be considered as a continuous function of closed-loop states, evaluated along the system's trajectories. With this under consideration, the remaining task is to design the control laws $v_x^*$ and $\omega^*$ to stabilize the origin for the error kinematics (\ref{eq:error-kinematics}), under the constraint that $v_y^*$ must satisfy \eqref{eq:constraint}. As we show below, $v_x^*$ and $\omega^*$ may be defined as functions of $(t,e_x)$ and $(t,e_\theta)$, respectively. 
\end{remark}

\subsubsection{Control of the error kinematics}

Using $v=v^*+\tilde{v}$ and $e:=[e_x~e_y~e_\theta]^\top$, the error system (\ref{eq:error-kinematics}) becomes
\begin{equation}\label{eq:error-system}
    \dot e = f(t,e,v^*) + g(t,e,\tilde{v})\tilde{v},
\end{equation}
where the drift of the nominal system is
\begin{equation}\label{eq:nominal}
    f(t,e,v^*):=\begin{bmatrix}
        \omega^* e_y+v_x^* -v_{xd}\cos e_\theta - v_{yd}\sin e_\theta\\
        -\omega^* e_x +v_y^*+v_{xd}\sin e_\theta  - v_{yd}\cos e_\theta  \\
        \omega^*-\omega_d
    \end{bmatrix}
\end{equation}
and the input-gain matrix $g$ is given by
\begin{equation}\label{eq:g}
    g(t,e,\tilde{v}):=
    \begin{bmatrix}
        1\ & 0 & \phantom{-}e_y\\
        0\ & 1 & -e_x\\
        0\ & 0 & \phantom{-}1
    \end{bmatrix}.
\end{equation}
The control goal for \eqref{eq:constraint} and \eqref{eq:error-system} is to stabilize,  uniformly, asymptotically and globally, an equilibrium containing $\{e=0\}$. To that end, we replace the latter in the right-hand side of \eqref{eq:nominal} to see that $f(t,0,v^*) = 0$ if and only if $v^*=v_d$, where $v_d := [v^*_{xd}\ v^*_{yd}\ \omega^*_d]^\top$. Therefore, it results sensible that for the purpose of analysis, we rewrite \eqref{eq:constraint} in terms of the velocity error $\bar{v}_y:=v_y^*-v_{yd}$, which is obtained by subtracting the second equation in (\ref{eq:5b}) from (\ref{eq:constraint}), to obtain
 \begin{equation}\label{eq:24}
     \dot{\bar{v}}_y=-\frac{d_{22}}{m_{22}}\bar{v}_y-\frac{m_{11}}{m_{22}}\big[ v_x^*(t,e_x)\omega^*(t,e_\theta) - v_{xd}(t)\omega_d(t)   \big].
 \end{equation}
The right-hand side of \eqref{eq:24} equals to zero if $\bar v_y =0$,  $v^*_x=v_{xd}$, and $\omega^* = \omega_d$. Therefore,  $v_x^*(t,e_x)$ and $\omega^*(t,e_\theta)$ must be designed to satisfy  $v_x^*(t,0) = v_{xd}$ and $\omega^*(t,0) = \omega_d$.

Now, based on the statement of Proposition \ref{prop:1}, we regard the system \eqref{eq:underactuated-closed-loop}, \eqref{eq:error-system}, and \eqref{eq:24} as a cascaded system, {\it i.e.},
\begin{subequations}
  \label{295}
  \begin{eqnarray}
     \label{295a} 
     \Sigma_1 &:& 
     \begin{bmatrix}
       \dot e \\ \dot{\bar v}_y
     \end{bmatrix}
     = 
     \begin{bmatrix}
       f(t,e,v^*(t,e,\bar v_y)) \\ 
       -\frac{d_{22}}{m_{22}}\bar{v}_y-\frac{m_{11}}{m_{22}}\big[ v_x^*\omega^* - v_{xd}\omega_d \big]
     \end{bmatrix}
     + 
     \begin{bmatrix}
         g(\,\cdot\,)\\ 0
     \end{bmatrix}\tilde v
\nonumber \\
      \\
     \label{295b} 
     \Sigma_2 &:&  M\dot{\tilde{v}} = - C(v(t))\tilde{v} - [D+K_d]\tilde{v} + \Phi(t,\tilde v) \qquad \quad 
  \end{eqnarray}
\end{subequations}
where $\Phi(t,\tilde v):= [0\ -m_{11} \tilde{v}_x\omega^*(t,e_\theta(t))\ 0]^\top$.

In \eqref{295}, the system $\Sigma_1$ is non-autonomous; its state variables are $e$ and $\bar v_y$;  and it is perturbed by $\tilde v$, which is the state of $\Sigma_2$. Indeed, note that $f$ in \eqref{eq:nominal} is a function of $t$ through the reference trajectories $v_{xd}(t)$ and $v_{yd}(t)$ and depends on $e$, also, through the functions $v_x^*(t,e_x)$ and $\omega^*(t,e_\theta)$. In addition, $f$ depends on the state $\bar v_y$\,---\,see the term $v_y^* = \bar v_y + v_{yd}(t)$ in the second element in \eqref{eq:nominal}. On the other hand, \eqref{295b} is that of another non-autonomous system with state $\tilde v$ and depends on time through the trajectories $v(t)$ and $e_\theta(t)$. Strictly speaking, these equations are well-posed provided that the solutions are forward complete \cite{loria2008feedback}. This is shown farther below.

Then, \cite[Theorem 2]{panteley1998global} may be invoked to assess uniform global asymptotic stability of the origin for the overall closed-loop system \eqref{eq:underactuated-closed-loop}, (\ref{eq:error-system}), and \eqref{eq:24}, provided that: 

\begin{enumerate}[label={(C\arabic*)},wide = 0pt, leftmargin = 2.3em]
    \item \label{A1}  Under given controls  $v_x^*(t,e)$ and $\omega^*(t,e_\theta)$ the origin $\{(e,\bar v_y) = (0,0)\}$ is UGAS for the nominal system  
\begin{equation}
  \label{326} \Sigma_{1\circ} : 
     \begin{bmatrix}
       \dot e \\ \dot{\bar v}_y
     \end{bmatrix}
     = 
     \begin{bmatrix}
       f(t,e,v^*(t,e,\bar v_y)) \\ 
       -\frac{d_{22}}{m_{22}}\bar{v}_y-\frac{m_{11}}{m_{22}}\big[ v_x^*\omega^* - v_{xd}\omega_d \big]
     \end{bmatrix}
\end{equation}
and we dispose of a polynomial non-strict Lyapunov function \cite[Proposition 4.8]{sepulchre1997constructive}. 
    
    \item The function $g$ in \eqref{eq:g} satisfies the condition of linear growth in $|e|$: 
    \begin{equation}
        |g(t,e,\tilde{v})|\le \alpha_1(|\tilde{v}|)|e|+\alpha_2(|\tilde{v}|),
    \end{equation}
    where $\alpha_1,\alpha_2:\mathbb{R}_{\ge 0}\to\mathbb{R}_{\ge 0}$ are continuous.  \label{A2}
    
    \item  The origin $\{\tilde{v}=0\}$ is UGES for (\ref{eq:underactuated-closed-loop}), uniformly in $e(t)$ and $v(t)$.  \label{A3}
\end{enumerate}
Condition \ref{A2} holds trivially while, as explained above, Condition \ref{A3} holds after Proposition \ref{prop:1} provided that $t\mapsto \omega^*(t, e_\theta(t))$ is uniformly bounded. For the latter and Condition \ref{A1} to hold, we introduce the use of the control laws 
\begin{subequations}
  \label{eq:22}
  \begin{eqnarray}
     v_x^*(t,e_x) &:=& -k_xe_x+v_{xd}(t) \label{eq:22a}\\
     \omega^*(t,e_\theta) &:= & -k_\theta \tanh(e_\theta)+\omega_d(t),\label{eq:22b}
  \end{eqnarray}
\end{subequations}
where $k_x$ and $k_\theta$ are positive design parameters to be determined later. Note that since  $\omega_d\in\mathcal{L}_\infty$ by assumption, we have $|\omega^*|_\infty\le k_\theta + \bar{\omega}_d=:\omega_M$, as required to invoke Proposition \ref{prop:1} in the analysis of \eqref{eq:underactuated-closed-loop} and (\ref{eq:error-system}).

\begin{remark}
The control laws in \eqref{eq:22} are reminiscent of those proposed (without saturation) in \cite{panteley1998exponential}, where cascades- and persistency-of-excitation-based tracking control of nonholonomic vehicles was originally proposed for the first time. See also the subsequent work \cite{lefeber2000tracking} and some of the references therein.
\end{remark}

It is left to analyze the stability of the nominal system $\Sigma_{1\circ}$ in \eqref{326}, {\it i.e.}, to verify Condition \ref{A1} above. To that end, we employ another cascades argument. Indeed, the system $\Sigma_{1\circ}$ may be regarded as an inner-loop cascade in which the $e_\theta$ dynamics is decoupled from the rest, {\it i.e.}, substituting (\ref{eq:22}) into (\ref{eq:nominal}), the system in \eqref{326} may be written in the form:
\vskip 5pt
\begin{tabular}{lc}
\hspace{-7mm}   $\Sigma_1' \,:$ &\hspace{-4mm} 
  \begin{minipage}{3mm}
    $\left\{ \
  \begin{minipage}{0.8\columnwidth}
   \begin{subequations}\label{eq:25}
     \begin{eqnarray} && \nonumber\\[-9mm]
    \dot{e}_p&=&
    A(t)e_p+\bar{\mathbf{v}}_y+B_1(t,e_p,e_\theta) \hspace{27.3mm}\label{eq:25a}\\
    \dot{\bar{v}}_y&=&-\frac{d_{22}}{m_{22}}\bar{v}_y+\frac{m_{11}}{m_{22}}k_x \omega_d(t) e_x + B_2(t,e_p,e_\theta)\nonumber \\ &&
\label{eq:25b}
     \end{eqnarray}
   \end{subequations}    
  \end{minipage}
    \right.$
  \end{minipage}
\end{tabular}
\begin{equation}
  \label{eq:25c} 
\mbox{}\hspace{-3mm} \Sigma_2' \,: \hspace{14pt} \dot{e}_\theta \,=\, -k_\theta \tanh(e_\theta) \hspace{37mm}\ 
\end{equation}
\vskip 4pt
\noindent where $e_p:=[e_x~e_y]^\top$, $\bar{\mathbf{v}}_y:=[0~\bar{v}_y]^\top$,
\begin{equation*}
    A(t) :=\begin{bmatrix}
        -k_x & \omega_d(t)\\
        - \omega_d(t) & 0
    \end{bmatrix},
\end{equation*}
\begin{equation*}
    B_1(\cdot):=\begin{bmatrix}
          v_{xd}(1-\cos (e_\theta))-v_{yd}\sin(e_\theta)-k_\theta e_y\tanh(e_\theta)\\
          v_{xd}\sin(e_\theta)+v_{yd}(1-\cos (e_\theta))+k_\theta e_x\tanh(e_\theta)
    \end{bmatrix},
\end{equation*}
and 
\begin{equation*}
    B_2(t,e_p,e_\theta) := \left[ v_{xd} -  k_x e_x \right] \frac{m_{11}}{m_{22}} k_\theta \tanh(e_\theta).
\end{equation*}

Following \cite[Theorem 2]{panteley1998global}, we conclude that for the system (\ref{eq:25})-(\ref{eq:25c}), the origin is UGAS if:
\begin{enumerate}[label={(C\arabic*')},wide = 0pt, leftmargin = 2.3em]
    \item The origin of the system $\Sigma_1'$ with  $\{e_\theta=0\}$, that is, 
      \begin{subequations}
        \label{eq:26}
        \begin{eqnarray}
        \Sigma_1'' \,: \dot{e}_p &=& A(t)e_p+\bar{\mathbf{v}}_y\label{eq:26a}\\
        \Sigma_2'' \,: \dot{\bar{v}}_y&=&-\frac{d_{22}}{m_{22}}\bar{v}_y+\frac{m_{11}}{m_{22}}k_x \omega_d(t) e_x,\quad\label{eq:26b}
        \end{eqnarray}
      \end{subequations}
    is UGAS with a polynomial Lyapunov function.\label{B1}
    
    \item The functions $B_1$ and $B_2$ satisfy the linear-growth condition
    \begin{equation}
        |B_i(t,e_p,e_\theta)|\le \alpha_{1i}(|e_\theta|)|e_p|+\alpha_{2i}(|e_\theta|),
    \end{equation}
    where $i\in \{1,2\}$, and  $\alpha_{1i}$ and $\alpha_{2i}$ are class $\mathcal K$ functions.\label{B2}
    
    \item The origin $\{e_\theta=0\}$ is GAS and LES for (\ref{eq:25c}).\label{B3}
\end{enumerate}

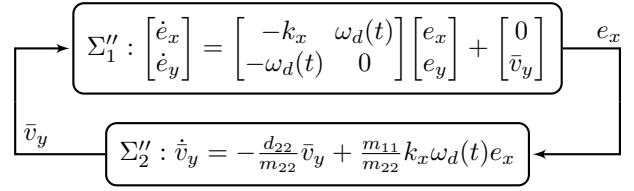
\begin{figure}[t]
    \centering
  \begin{tikzpicture}[node distance=1.5cm,auto,>=latex']
    \node [mynode] (sigmaone) at (0,.7) {
      $\Sigma_1'': \begin{bmatrix}
        \dot e_x \\ \dot e_y
      \end{bmatrix} = 
      \begin{bmatrix}
        -k_x & \omega_d(t) \\
        -\omega_d(t) & 0 
      \end{bmatrix}\!
      \begin{bmatrix}
        e_x \\ e_y
      \end{bmatrix} + 
      \begin{bmatrix}
        0 \\ \bar v_y
      \end{bmatrix}$} ;
    \node [mynode] (sigmatwo) at (0,-.7) {$\Sigma_2'':\dot{\bar{v}}_y = -\frac{d_{22}}{m_{22}}\bar{v}_y+\frac{m_{11}}{m_{22}}k_x \omega_d(t) e_x$};
    \coordinate (a) at (4,0) {};
    \coordinate (b) at (-4,0) {};

    \draw[thick] (sigmaone.east) -| (a) node [label={[xshift=-4pt, yshift=15pt]{\color{black}$e_x$}}] {} ;
    \draw[myarrow] (a) |- (sigmatwo.east) ;

    \draw[thick] (sigmatwo.west) -| (b) node [label={[xshift=8pt, yshift=-25pt]{\color{black}$\bar v_y$}}] {} ;
    \draw[myarrow] (b) |- (sigmaone.west) ;
\end{tikzpicture}

  \caption{Small-gain feedback representation of the system (\ref{eq:26}).}
    \label{fig:feedback}
\end{figure}

Conditions \ref{B2} and \ref{B3} hold trivially. For Condition \ref{B1}, the system (\ref{eq:26}) may be regarded as the feedback interconnection of two subsystems, as shown in Fig. \ref{fig:feedback}.  The stability of the origin of the feedback interconnection may be concluded by invoking the small-gain theorem \cite[Theorem 8.2.1]{schaft2017l2}, which is recalled in the appendix for convenience. 

\begin{proposition}[Nominal system]\label{prop:2}
    Consider the feedback-interconnected system (\ref{eq:26}). Suppose that $\omega_d$ is persistently exciting, {\it i.e.}, there exist $\mu,T>0$ such that 
    \begin{equation}\label{eq:PE}
    \int_t^{t+T}\omega_d(s)^2{\rm d}s\ge \mu,\quad \forall~t\ge 0.
    \end{equation}
    Furthermore, suppose that 
    \begin{equation}\label{eq:bound}
\max\{|\omega_d|_\infty,|\dot{\omega}_d|_\infty\}\leq
\bar{\omega}_d\le  \left(\frac{d_{22}}{m_{11}}\cdot\frac{\mu^2}{4\mu T + T^2}\right)^{\frac{1}{3}}.
    \end{equation}
    Then there exists a constant $\bar{k}_x>0$ such that for all $k_x\in(0,\bar{k}_x]$, the origin of the feedback interconnection (\ref{eq:26}) is UGAS.
\end{proposition}
\begin{proof}
  The proof follows by invoking Theorem \ref{thm:small-gain} in the Appendix. In accordance, we  establish  the following:
\begin{enumerate}
    \item[1)] $\Sigma_1''$ has $\mathcal{L}_2$-gain $\gamma(\Sigma_1'')\le \gamma_1$;
    \item[2)] $\Sigma_2''$ has $\mathcal{L}_2$-gain $\gamma(\Sigma_2'')\le \gamma_2$;
    \item[3)] The small-gain condition $\gamma_1 \cdot \gamma_2<1$ holds.
\end{enumerate}
\textbf{Item 1)} We prove item 1) by constructing an ISS Lyapunov function for $\Sigma_1''$. To that end, for the linear system $\Sigma_{1\circ}'':\dot{e}_p=A(t)e_p$, we show that 
\begin{equation}\label{eq:W1}
    W_1(t,e_p):=\frac{1}{2}[\Upsilon_{\omega_d^2}(t)+\alpha](e_x^2+e_y^2)-\omega_d(t)e_xe_y,
\end{equation}
 where $\alpha$ is a constant satisfying
\begin{equation}\label{eq:alpha}
    \alpha\ge \max\left\{\bar{\omega}_d,\frac{2\bar{\omega}_d^2}{k_x}\left(1+\frac{T}{4\mu}(1+k_x)^2\right)\right\},
\end{equation}
is a strict Lyapunov function. The function $\Upsilon: \mathbb{R}_{\ge 0}\to\mathbb{R}_{\ge 0}$, which is inspired from \cite{MAGLOR-JP1}, is defined as
\begin{equation}
    \Upsilon_{\omega_d^2}(t):=1+2\bar{\omega}_d^2T-\frac{2}{T}\int_t^{t+T}\int_t^m\omega_d^2(s){\rm d}s{\rm d}m,
\end{equation}
and it satisfies
\begin{equation}\label{eq:Upsilon}
    1+\bar{\omega}_d^2T\le  \Upsilon_{\omega_d^2}(t)\le 1+2\bar{\omega}_d^2T,  \qquad \forall\,t\ge 0.
\end{equation}
It follows from (\ref{eq:alpha}), (\ref{eq:Upsilon}), and  Young’s inequality, that
\begin{equation}
    \frac{1}{2}|e_p|^2\le W_1(t,e_p)\le \frac{1}{2}[1+2\bar{\omega}_d^2T+\alpha+\bar{\omega}_d]|e_p|^2.
\end{equation}

Now, the time derivative of $\Upsilon_{\omega_d^2}(t)$ satisfies 
\begin{align}
    \dot{\Upsilon}_{\omega_d^2}(t)&=-\frac{2}{T}\int_t^{t+T} \omega_d^2(s){\rm d}s+2\omega_d^2(t)  \notag\\
    &\le -\frac{2\mu}{T}+2\omega_d^2(t),
\end{align}
so the total derivative of $W_1$ along trajectories of $\Sigma_{1\circ}''$ yields
\begin{align}
    \dot{W}_1|_{\Sigma_{1\circ}''}&=-k_x[\Upsilon_{\omega_d^2}(t)+\alpha]e_x^2-\left(\frac{1}{T}\int_t^{t+T}\omega_d^2(s){\rm d}s\right)|e_p|^2  \notag\\
    &\quad +2\omega_d^2(t)e_x^2+[k_x\omega_d(t)-\dot{\omega}_d(t)]e_xe_y.
\end{align}
It follows from (\ref{eq:alpha}) and Young's inequality that
\begin{equation}
    \dot{W}_1|_{\Sigma_{1\circ}''}\le -\frac{\mu}{2T}|e_p|^2.
\end{equation}
Hence, for the linear system $\dot{e}_p=A(t)e_p$, $W_1$ is a strict Lyapunov function, and the origin of the system is UGES.

Next, we show that $W_1$ is an ISS Lyapunov function for $\Sigma_1''$. The total derivative of $W_1$ along trajectories of $\Sigma_1''$ yields
\begin{align}
    \dot{W}_1 
    &\le -\frac{\mu}{2T}|e_p|^2 + \beta|e_p||\bar{v}_y|, \notag
\end{align}
\color{black}where $\beta:=1+2\bar{\omega}_d^2T+\alpha+\bar{\omega}_d$. It follows from Young's inequality that 
\begin{equation}
    \dot{W}_1\le -\frac{\mu}{4T}|e_p|^2 + \frac{T\beta^2}{\mu} \bar{v}_y^2.
\end{equation}
Let $V_1:=\frac{2T}{\mu}W_1$ and we have
\begin{equation}
    \dot{V}_1\le -\frac{1}{2}|e_p|^2+\frac{\gamma_1^2}{2}\bar{v}_y^2,
\end{equation}
where 
\begin{equation}\label{eq:gamma1}
    \gamma_1:=\frac{2T\beta}{\mu}.
\end{equation}
It follows that $\Sigma_1''$ has $\mathcal{L}_2$-gain $\gamma(\Sigma_1'')\le \gamma_1$. 

\textbf{Item 2)} The total derivative of the storage function $V_2:=\frac{m_{22}}{2d_{22}}\bar{v}_y^2$ along trajectories of $\Sigma_2''$ yields
\begin{equation}
    \dot{V}_2   \leq  -\frac{1}{2}\bar{v}_y^2+\frac{\gamma_2^2}{2} e_x^2,\label{eq:dot_V2}
\end{equation}
\color{black}where 
\begin{equation}\label{eq:gamma2}
    \gamma_2:=\frac{m_{11}k_x\bar{\omega}_d}{d_{22}}.
\end{equation}
It follows that $\Sigma_2''$ has $\mathcal{L}_2$-gain $\gamma(\Sigma_2'')\le \gamma_2$. Note that $\gamma_2$ can be made arbitrarily small by selecting the parameter $k_x$.

\textbf{Item 3)} From (\ref{eq:gamma1}) and (\ref{eq:gamma2}), we have 
\begin{align}
    \gamma_1\cdot\gamma_2&= \frac{2T\beta}{\mu}\cdot\frac{m_{11}}{d_{22}}k_x\bar{\omega}_d \notag\\
    &<\frac{2Tm_{11}}{\mu d_{22}}k_x\bar{\omega}_d\bigg(1+2\bar{\omega}_d^2T+2\bar{\omega}_d \notag\\
    &\quad~ +\left.\frac{2\bar{\omega}_d^2}{k_x}\left(1+\frac{T}{4\mu}(1+k_x^2)\right) \right).  \label{eq:above}
\end{align}
We see that the right-hand side of the inequality (\ref{eq:above}) is a continuous function of $k_x$, and it monotonically decreases as $k_x$ decreases. Therefore, taking the limit as $k_x$ tends to zero yields
\begin{equation}
    \gamma_1\cdot\gamma_2<\frac{4Tm_{11}\bar{\omega}_d^3}{\mu d_{22}}\left(1+\frac{T}{4\mu}\right).
\end{equation}
Finally, $\bar{\omega}_d$ satisfies (\ref{eq:bound}), which implies that the small-gain condition $\gamma_1\cdot\gamma_2<1$ holds. The zero-state detectability conditions for $\Sigma_1''$ and $\Sigma_2''$ can be easily verified. It follows from Theorem \ref{thm:small-gain} that the origin of the feedback interconnection (\ref{eq:26a})-(\ref{eq:26b}) is UGAS, which completes the proof.    
\end{proof}

The proof of Proposition \ref{prop:2} not only demonstrates the UGAS of the feedback interconnection (\ref{eq:26}) but also provides a strict Lyapunov function. Because $\gamma_1\cdot\gamma_2<1$, there exists $\lambda>0$ such that $\gamma_1<\lambda<\frac{1}{\gamma_2}$. Then, $\mathcal{V}:=V_1+\lambda^2V_2$ is a strict Lyapunov function for (\ref{eq:26}). Due to the fact that $\mathcal{V}$ is quadratic, we verify Condition \ref{B1}, and thus, the origin of the system (\ref{eq:25})-(\ref{eq:25c}) is UGAS. Consequently, we verify Condition \ref{A1} and conclude that the origin of the closed-loop system (\ref{eq:error-system}) with control laws (\ref{eq:constraint}) and (\ref{eq:22}) is UGAS \cite[Theorem 2]{panteley1998global}. Thus, we have the following. 
\begin{proposition}[Trajectory tracking]\label{thm:1}
    Consider the system (\ref{eq:1}) together with with the control laws (\ref{eq:constraint}), (\ref{eq:underactuated-control}), and (\ref{eq:22}). Suppose that there exists a positive constant $\bar{\omega}_d$ such that $\max\{|\omega_d|_\infty,|\dot{\omega}_d|_\infty\}\le \bar{\omega}_d$. Furthermore, suppose that $\omega_d$ is persistently exciting in the sense of (\ref{eq:PE}), and that $\bar{\omega}_d$ satisfies (\ref{eq:bound}). Then, there exists a constant $\bar{k}_x>0$ such that for all $k_x\in(0,\bar{k}_x]$ and all $k_\theta\in(0,\infty)$, there exists a constant $\bar{k}_{dx}>0$, such that for all $k_{dx}\in[\bar{k}_{dx},\infty)$ and all $k_{d\omega}\in (0,\infty)$, the origin of the closed-loop system, $(e,\bar{v}_y,\tilde{v})=(0,0,0)$, is UGAS (in particular, the trajectory tracking control problem is solved).
\end{proposition}
\begin{proof}
Consider the function 
\begin{equation}
  \label{614} U(e,\bar v_y, \tilde v) = |e|^2 + |\tilde v|^2 + |\bar v_y|^2
\end{equation}
After direct computation of $\dot U$ along the trajectories of \eqref{eq:underactuated-closed-loop}, \eqref{eq:error-system}, and \eqref{eq:24}, we see that there exist $a$ and $b>0$ such that $\dot \nu(t) \leq a\nu(t) + b$, where $\nu(t):= U(e(t),\bar v_y(t), \tilde v(t))$. Therefore, the closed-loop system is forward complete. After Propositions \ref{prop:1} and \ref{prop:2}, and the argumentation on \eqref{295} above, the statement follows.
\end{proof}

\subsection{Formation Tracking control}

Let us consider a group of $n$ surface vessels that are required to advance in formation and are modeled by the equations
\begin{subequations}
\label{eq:46}
  \begin{eqnarray}
    &\dot{q}_i=J(q_i)v_i& \label{eq:46a}\\
    &M\dot{v}_i+C(v_i)v_i+Dv_i=G\tau_i,& \label{eq:46b}
  \end{eqnarray}
\end{subequations}
where for the $i$th vessel, $q_i=[x_i~y_i~\theta_i]^\top$ is the configuration containing the Cartesian coordinates $(x_i,y_i)$ and the orientation $\theta_i$ in the fixed inertia frame; $v_i=[v_{xi}~v_{yi}~\omega_i]^\top$ is the generalized velocity vector consisting of the linear velocity $(v_{xi},v_{yi})$ and the angular velocity $\omega_i$ in the body-fixed frame; and $\tau_i=[\tau_{xi}~\tau_{\omega i}]^\top$ is the control input vector consisting of the surge force and the yaw torque, respectively.

The formation tracking control problem consists in making the $n$ surface vessels take specific postures assigned by the topology designer while guiding the swarm along a path determined by a virtual reference surface vessel labeled as agent-$0$. It is assumed that the $i$th surface vessel follows a leader, indexed $i-1$, thus forming a directed spanning tree graph communication topology. In other words, each surface vessel has only one leader, but it may have several followers. Additionally, the swarm has only one swarm leader vessel labeled as agent-$1$, which is the sole entity with access to the reference trajectory generated by the virtual leader agent-$0$.

Similar to (\ref{eq:error}), we define formation error $e_i=[e_{xi}~ e_{yi}~ e_{\theta i}]^\top$ for the $i$th surface vessel in its local coordinates as 
\begin{align}\label{eq:formation-error}
    \begin{bmatrix}
        e_{xi} \\ e_{yi} \\ e_{\theta i}
    \end{bmatrix}
    =
    \begin{bmatrix}
        \cos\theta_i & \sin\theta_i & 0 \\
        -\sin\theta_i & \cos\theta_i & 0\\
        0 & 0 & 1
    \end{bmatrix}
    \begin{bmatrix}
        x_i-x_{i-1}-d_{xi} \\ y_i-y_{i-1}-d_{yi} \\ \theta_i-\theta_{i-1}
    \end{bmatrix},
\end{align}
where $(d_{xi},d_{yi})$ is the desired relative displacement between the $i$th vessel and its leader, and defines the geometry of the formation. Thus, the error dynamics between the $i$th surface vessel and its leader become
\begin{subequations}
\label{eq:error-kinematics-formation}
  \begin{eqnarray}
    \label{eq:error-kinematics-formation:a}
    \dot{e}_{xi}&=&\omega_i e_{yi} +v_{xi} -v_{x(i-1)}\cos e_{\theta i} - v_{y(i-1)}\sin e_{\theta i} \\
    \label{eq:error-kinematics-formation:b}
    \dot{e}_{yi}&=&-\omega_i e_{xi} +v_{yi} +v_{x(i-1)}\sin e_{\theta i} - v_{y(i-1)}\cos e_{\theta i}\qquad\ \\
    \label{eq:error-kinematics-formation:c}
    \dot{e}_{\theta i}&=&\omega_i-\omega_{i-1}.
  \end{eqnarray}
\end{subequations}

For $i=1$, we recover the error dynamics for the case of trajectory tracking, {\it i.e.}, $v_0:=v_d$. Then, for $i\le n$, similar to (\ref{eq:constraint}) and (\ref{eq:22}), we introduce the virtual controls 
\begin{subequations}\label{eq:49}
  \begin{eqnarray}
      v_{xi}^*&:=&-k_{xi}e_{xi}+v_{x(i-1)} \label{eq:49a}\\
    \omega_i^*&:=&-k_{\theta i} \tanh(e_{\theta i})+\omega_{i-1},\label{eq:49b}
  \end{eqnarray}
\end{subequations}
and $v_{yi}^*$ given by the integration of the differential equation
\begin{equation}\label{eq:mas-constraint}
    \dot{v}_{yi}^*=-\frac{m_{11}}{m_{22}}v_{xi}^*\omega_i^*-\frac{d_{22}}{m_{22}}v_{yi}^*, \quad v_{yi}^*(0)=v_{y(i-1)}(0).
\end{equation}
With $v_i^*=[v_{xi}^*~v_{yi}^*~\omega_i^*]^\top$, the actual control $\tau_i$ is given by
\begin{equation}\label{eq:mas-underactuated-control}
    \tau_i=G^\dagger(M\dot{v}_i^*+C(v_i)v_i^*+Dv_i^*-K_{di} \tilde{v}_i),
\end{equation}
where the velocity error $\tilde{v}_i:=[\tilde{v}_{xi}~\tilde{v}_{yi}~\tilde{\omega}_i]^\top:=v_i-v_i^*$ and the gain matrix $K_{di}=\operatorname{diag}\{k_{dxi},0,k_{d\omega i}\}\ge 0$.

Next, let us define $\Delta v_i:=[\Delta v_{xi}~\Delta v_{yi}~ \Delta \omega_i]^\top:=v_i-v_d$ and $\bar{v}_{yi}:=v_{yi}^*-v_{y(i-1)}$ for all $i\le n$. Then, we replace $v_i$ with $v_i^*+\tilde{v}_i$ in (\ref{eq:error-kinematics-formation}), substitute (\ref{eq:49}) into (\ref{eq:error-kinematics-formation}) and (\ref{eq:mas-constraint}), and substitute (\ref{eq:mas-underactuated-control}) into (\ref{eq:46b}), which yields the closed-loop error system 
\begin{subequations}\label{eq:CL}
\begin{eqnarray}
  \dot{\xi}_i&=&F(t,\xi_i)+H(e_i)\Delta v_{i-1}+G(e_i)\tilde{v}_i \label{eq:CL-kinematics}\\
  M\dot{\tilde{v}}_i&=&-C(v_i(t))\tilde{v}_i-[D+K_{di}]\tilde{v}_i+\Phi_i(t,\tilde{v}_i),\qquad\;\; \label{eq:CL-dynamics}
\end{eqnarray}
\end{subequations}
where---\textit{cf.} (\ref{295}), $\xi_i=[e_{pi}^\top~\bar{v}_{yi}~e_{\theta i}]^\top$, $e_{pi}=[e_{xi}~e_{yi}]^\top$,
\begin{equation*}
    F(t,\xi_i):=\begin{bmatrix}
        A(t)e_{pi}+\bar{\mathbf{v}}_{yi}+B_1(t,e_{pi},e_{\theta i})\\
        -\frac{d_{22}}{m_{22}}\bar{v}_{yi}+\frac{m_{11}}{m_{22}}k_{xi}\omega_d(t) e_{xi}+B_2(t,e_{pi},e_{\theta i})\\
        -k_{\theta i} \tanh{e_{\theta i}}
    \end{bmatrix},
\end{equation*}
\begin{equation*}
    H(e_i):=\begin{bmatrix}
        1-\cos e_{\theta i} & -\sin e_{\theta i} & e_{yi}\\
        \sin e_{\theta i} & 1-\cos e_{\theta i} & -e_{xi}\\
        \frac{m_{11}}{m_{22}} k_{\theta i}\tanh{e_{\theta i}}& 0 & \frac{m_{11}}{m_{22}}k_{xi}e_{xi}\\
        0 & 0 & 0
    \end{bmatrix},
\end{equation*}
\begin{equation*}
    G(e_i):=\begin{bmatrix}
        1 & 0 & e_{yi}\\
        0 & 1 & -e_{xi}\\
        0 & 0 & 0\\
        0 & 0 & 1
    \end{bmatrix},
\end{equation*}
$\bar{\mathbf{v}}_{yi}:=[0~\bar{v}_{yi}]^\top$, $\Phi_i(t,\tilde{v}):=[0~-m_{11}\tilde{v}_{xi}\omega_i^*(t,e_{\theta i}(t))~0]^\top$, and matrices $A(\cdot)$, $B_1(\cdot)$, $B_2(\cdot)$ are defined in (\ref{eq:25})-(\ref{eq:25c}). The overall closed-loop multi-agent kinematics system has the cascaded form:
\begin{subequations}
\label{eq:54}
  \begin{eqnarray}
    \dot{\xi}_n&=&F(t,\xi_n)+H(e_n)\Delta v_{n-1}+G(e_n)\tilde{v}_n \label{eq:54a}\\
    &\vdots \notag\\
    \dot{\xi}_2&=&F(t,\xi_2)+H(e_2)\Delta v_{1}+G(e_2)\tilde{v}_2\label{eq:54b}\\
    \dot{\xi}_1&=&F(t,\xi_1)+G(e_1)\tilde{v}_1.\label{eq:54c}
  \end{eqnarray}
\end{subequations}

We propose our second main result as follows.

\begin{proposition}[Formation tracking]
    Consider the Euler-Lagrange systems (\ref{eq:46}) in closed loop with the control laws  (\ref{eq:49})-(\ref{eq:mas-underactuated-control}), for each $i\le n$. Define $v_0:=v_d$ and suppose that there exists a positive constant $\bar{\omega}_d$ such that $\max\{|\omega_d|_\infty,|\dot{\omega}_d|_\infty\}\le \bar{\omega}_d$. Furthermore, suppose that $\omega_d$ is persistently exciting in the sense of (\ref{eq:PE}), and that $\bar{\omega}_d$ satisfies (\ref{eq:bound}). Then, there exist constants $\bar{k}_x$ and $\bar{k}_{dx}>0$ such that, if for each $i\le n$,  $k_{xi}\in(0,\bar{k}_x]$, $k_{dxi}\ge \bar{k}_{dx}$, $k_{\theta i} > 0$, and  $k_{d\omega i}\ge 0$, the origin of the closed-loop multi-agent system, $(\xi,\tilde{v})=(0,0)$, where $\xi:=[\xi_1^\top\ldots\xi_n^\top]^\top$ and $\tilde{v}=[\tilde{v}_1^\top\ldots\tilde{v}_n^\top]^\top$, is UGAS (consequently, the formation tracking control problem is solved). 
\end{proposition}
\begin{proof} The proof is established by induction. Let $i=1$. Then, the closed-loop system is composed of (\ref{eq:54c}) and (\ref{eq:CL-dynamics}) and has a nested cascaded form. It follows from Proposition \ref{prop:1} that the origin of (\ref{eq:CL-dynamics}) with $i=1$ is UGES, so $\tilde{v}_1\to 0$ exponentially. Then, on the set $\{\tilde{v}_1=0\}$, the nominal system $\dot \xi_1=F(t,\xi_1)$ corresponds to (\ref{eq:25})-(\ref{eq:25c}), which also has a cascaded structure. Following Proposition \ref{prop:2} and recursively using the cascades argument, we conclude that the origin $(\xi_1,\tilde{v}_1)=(0,0)$ is UGAS. Furthermore, it follows from the proof of Proposition \ref{prop:2} that $\xi_1\to 0$ exponentially near the origin due to the quadratic Lyapunov function $\mathcal{V}$. By virtue of the control design, $\tilde{v}_1\to 0$ and $\xi_1\to 0$ exponentially near the origin implies that $\Delta v_1\to 0$ exponentially near the origin. Now, let us rewrite the closed-loop system (\ref{eq:54c}) and (\ref{eq:CL-dynamics}) with $i=1$ as $\dot{\zeta}_1=F_{{\rm cl},1}(t,\zeta_1)$, where $\zeta_1:=[\xi_1^\top~\tilde{v}_1^\top]^\top$.

Next, let $i=2$ and consider the closed-loop system 
\begin{subequations} \label{eq:55}
  \begin{eqnarray}
    \dot{\xi}_2&=&F(t,\xi_2)+H(e_2)\Delta v_{1}+G(e_2)\tilde{v}_2 \label{eq:55a}\\
    \dot{\zeta}_1&=&F_{{\rm cl},1}(t,\zeta_1) \label{eq:55b}\\
    M\dot{\tilde{v}}_2&=&-C(v_2)\tilde{v}_2-[D+K_{d2}]\tilde{v}_2+\Phi_2(t,\tilde{v}_2).\qquad\label{eq:55c}
  \end{eqnarray}
\end{subequations}
Note that the system (\ref{eq:55}) also has a nested cascaded form. From  Proposition \ref{prop:1}, it follows that the origin of (\ref{eq:55c}) is UGES, so $\tilde{v}_2\to 0$ exponentially. Also, the origin of (\ref{eq:55b}) is UGAS, and $\zeta_1\to0$ exponentially near the origin. Therefore, to analyze the stability of the origin for (\ref{eq:55}), it is enough to establish UGAS for the system $\dot{\xi}_2=F(t,\xi_2)$ and invoke again \cite[Theorem 2]{panteley1998global}, because the matrices $H(e_2)$ and $G(e_2)$ satisfy the condition of linear growth in  $|e_2|$. To that end, we note that the system $\dot{\xi}_2=F(t,\xi_2)$ also has a cascaded structure, as in equations (\ref{eq:25})-(\ref{eq:25c}). Therefore, following Proposition \ref{prop:2} again, we conclude that $(\xi_2,\zeta_1,\tilde{v}_2)=(0,0,0)$ is UGAS for (\ref{eq:55}). Next, we rewrite the closed-loop system (\ref{eq:55}) as $\dot{\zeta}_2=F_{{\rm cl},2}(t,\zeta_2)$, where $\zeta_2:=[\xi_2^\top~\zeta_1^\top~\tilde{v}_2^\top]^\top$. For $i=3$, the closed-loop system has exactly the same form as \eqref{eq:55}, with the index `${}_2$' replaced by `${}_3$' and 
all of the previous arguments hold. The result follows by induction.
\end{proof}

\section{Simulation Results}\label{sec:simulation}

We consider a group of four surface vessels in the linear communication topology following a virtual leader. That is, each surface vessel has only one follower except for agent-4. Each surface vessel is considered to be modeled by 
\begin{subequations}\label{eq:sim}
\color{myred}{
  \begin{eqnarray}
    &\dot{q}_i=J(q_i)v_i& \label{eq:sim-a}\\
    &M\dot{v}_i+C(v_i)v_i+Dv_i=G\tau_i + w_i(t),& \label{eq:sim-b}
  \end{eqnarray}
  }
\end{subequations}
\hspace{-0.22cm}with \textit{nominal} parameters $m_{11}=1.012$, $m_{22}=1.982$, $m_{33}=0.354$, $d_{11}=3.436$, $d_{22}=18.99$, and $d_{33}=0.864$, where the parameters are from a laboratory-size surface vessel. All parameters are given in SI units. \myred{In \eqref{eq:sim}, \( w_i(t) \in \mathbb{R}^3 \) represents external disturbances and is modeled as white noise.} The desired geometric formation shape of the four surface vessels is a diamond configuration, {\it i.e.}, $(d_{x1},d_{y1})=(0,0)$, $(d_{x2},d_{y2})=(-2,-2)$, $(d_{x3},d_{y3})=(4,0)$, and $(d_{x4},d_{y4})=(-2,-2)$. The reference trajectory is generated by the virtual leader with input $(\tau_{x0},\tau_{\omega0})=(2,0.3\sin({0.3t}))$ and with zero initial conditions. The initial conditions of the follower agents are randomly generated as $q_1(0)=[1.46,0.45,1.33]^\top$, $q_2(0)=[-3.45,2.25,1.02]^\top$, $q_3(0)=[-5.63,-4.94,-0.18]^\top$, and $q_4(0)=[-1.17,-5.23,-0.78]^\top$. \myred{Based on the reference trajectory, the upper bound of the reference angular velocity is given by \(\bar{\omega}_d = 0.355\). The PE parameters are \(\mu = 1.32\) and \(T = 20.94\); so Inequality \eqref{eq:bound} holds.} The control gains are set to $k_{xi}=k_{\theta i}=0.2$, and $K_{di}=\operatorname{diag}\{10,0,10\}$. \myred{With the chosen control gains, one can verify that the right-hand side of \eqref{eq:above} is less than \(0.982\), ensuring that the small-gain condition holds.} The simulation results \myred{with $w_i(t)\equiv 0$} are shown in Figs. \ref{fig:path} and \ref{fig:error} (dash-dotted lines), which illustrate the physical paths described by the swarm and the formation tracking error trajectories for each leader-follower pair. As shown in Fig. \ref{fig:error} (dash-dotted lines), the formation tracking errors converge to zero asymptotically.

\color{myred}To illustrate the local input-to-state stability property (robustness), guaranteed by the UGAS property achieved under the proposed approach statistically, \color{black} Fig. \ref{fig:error} provides the evolution of the formation tracking errors with \textit{\color{myred}parameter uncertainties} and \textit{\color{myred}input disturbances} over 100 runs with random initial states, {\it i.e.}, $(x_i,y_i,v_i)\sim \mathcal{N}(0,5^2)$ and $\theta_i\sim \mathcal{U}(-\pi,\pi)$. Each model parameter of the four follower vessels is randomly generated in each run by multiplying a random variable factor $\eta\sim\mathcal{U}(0.5,1.5)$ to the nominal parameter. \myred{In other words, up to $\pm 50\%$ parameter uncertainties are considered in the simulations. Moreover, white noise disturbances $w_i(t)$ are added to the surface vessel dynamics with noise power 0.1 and sample time 0.01.} The simulation results of the 100 runs are shown in Fig. \ref{fig:error}. The results indicate that all formation tracking error trajectories converge to and remain within a small neighborhood around zero, \myred{even under parameter uncertainties and disturbances}.

\begin{figure}[t]
    \centering
    \includegraphics[scale=0.34]{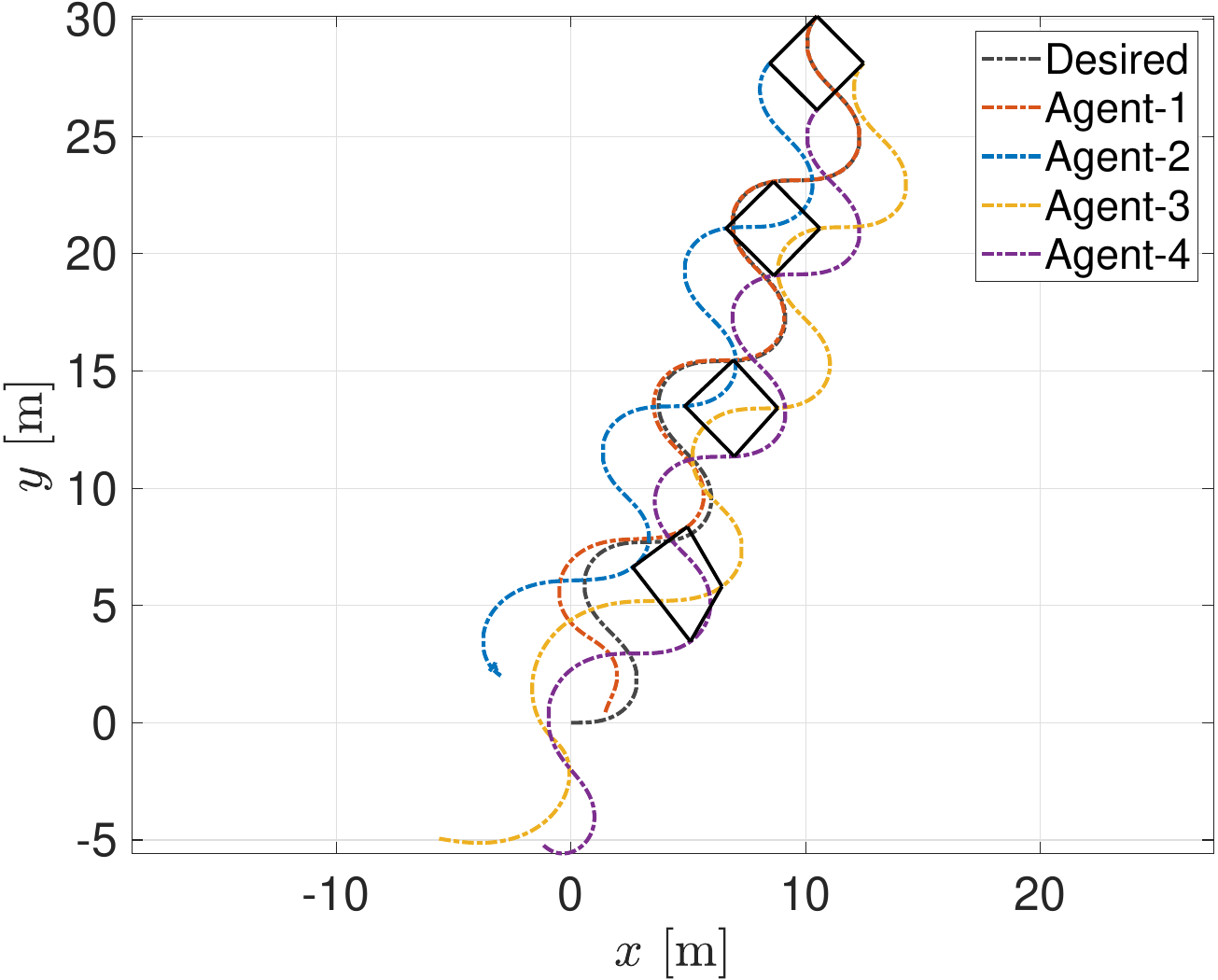}
    \caption{\color{myred}Illustration of the paths in formation tracking.}
    \label{fig:path}
\end{figure}

\begin{figure}[t]
    \centering
    \includegraphics[scale=0.34]{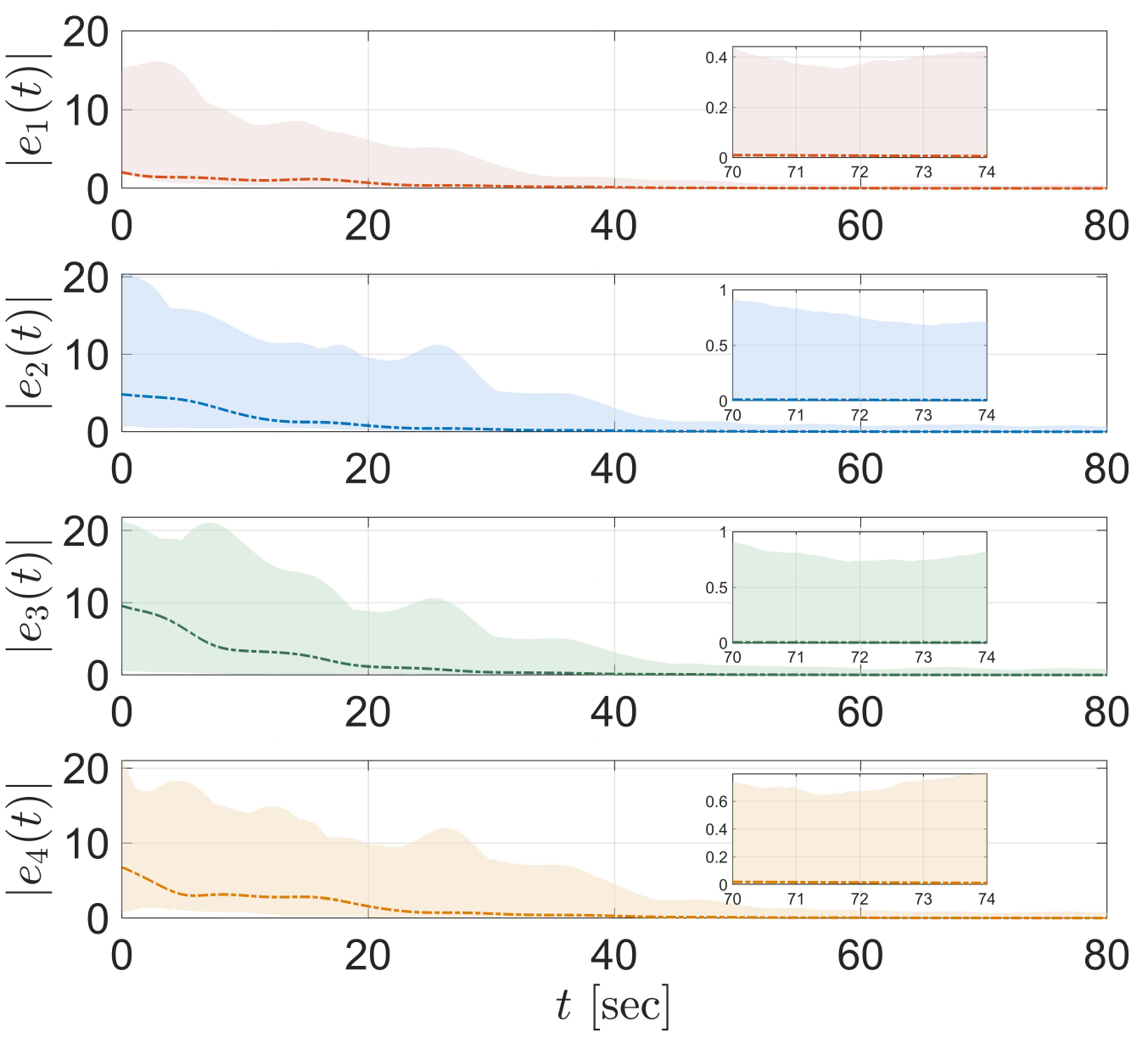}
    \caption{\color{myred}Convergence of the relative errors (in norm) for each leader-follower pair. The shadowed colored regions represent the envelopes containing the trajectories of 100 simulation tests with randomly chosen initial conditions and under the influence of \textit{parameter uncertainties} and \textit{external disturbances}.}
    \label{fig:error}
\end{figure}

\section{Conclusion}\label{sec:conclusion}

We present a simple formation tracking controller for underactuated surface vessels equipped with only two propellers. The control laws at the kinematics level are \textit{linear} and \textit{saturated linear} and exhibit a low-gain feature. At the kinetics level, our approach extends the \textit{classical PD+ controller} to underactuated vehicle systems, ensuring uniform global asymptotic stability for the closed-loop system. {\color{myred}Future research will focus on extending the presented approach to underactuated surface vessels with more complex models and under relaxed assumptions, such as switching topologies, communication delays, and heterogeneous networks.}

\section*{Appendix: $\mathcal{L}_2$-Gain and the Small-Gain Theorem}

 We say that the $\mathcal{L}_2$-gain of the dynamical system 
\begin{equation*}
    \Sigma:\left\{\begin{aligned}
        \dot{x}&=f(x,u),\quad x\in\mathbb{R}^n,u\in\mathbb{R}^m\\
        y&=h(x,u),\quad y\in\mathbb{R}^p
    \end{aligned}\right.
\end{equation*}
is less than or equal to $\gamma$ if it is dissipative with respect to the supply rate $s(u,y):=\frac{1}{2}\gamma^2|u|^2-\frac{1}{2}|y|^2$. That is, there exists a continuously differentiable storage function $V:\mathbb{R}^n\to \mathbb{R}_{\ge 0}$ such that 
\begin{equation*}
    \dot{V}:=\frac{\partial V}{\partial x}(x)f(x,u)\le \frac{\gamma^2}{2}|u|^2-\frac{1}{2}|h(x,u)|^2,\quad \forall~x,u.
\end{equation*}
The $\mathcal{L}_2$-gain of $\Sigma$ is defined as 
\begin{equation*}
    \gamma(\Sigma):=\inf\{\gamma:\Sigma~\text{has}~\mathcal{L}_2\text{-gain}~\le \gamma\}.
\end{equation*}

Consider the feedback interconnection of two systems
\begin{equation*}
    \Sigma_i:\left\{\begin{aligned}
        \dot{x}_i&=f_i(x_i,u_i)\\
        y_i&=h_i(x_i,u_i)
    \end{aligned}\right.,\quad i=1,2,
\end{equation*}
where $u_1=-y_2$ and $u_2=y_1$. Denote the storage function of $\Sigma_1$ and $\Sigma_2$ by $V_1$ and $V_2$, respectively.

\begin{theorem}[Small-gain theorem \cite{schaft2017l2}]\label{thm:small-gain}
    Suppose that the systems $\Sigma_1$ and $\Sigma_2$ have $\mathcal{L}_2$-gains such that $\gamma(\Sigma_1)\le \gamma_1$ and $\gamma(\Sigma_2)\le \gamma_2$, with $\gamma_1\cdot \gamma_2<1$. Suppose that two functions $V_1$ and $V_2$ satisfying
    \begin{equation*}
    \dot{V}_i:=\frac{\partial V_i}{\partial x_i}(x_i)f_i(x_i,u_i)\le \frac{\gamma_i^2}{2}|u_i|^2-\frac{1}{2}|y_i|^2,\quad i=1,2,
\end{equation*}
are proper and have global minima at $(x_1,x_2)=(0,0)$, and suppose that $\Sigma_1$ and $\Sigma_2$ are zero-state detectable. Then the origin of the feedback interconnection $(x_1,x_2)=(0,0)$ is globally asymptotically stable.
\end{theorem}

\vskip 10pt
\section*{References}
\vspace{-0.5cm}
\bibliographystyle{ieeetr} 
\bibliography{v2_main}   

\end{document}